\documentclass[submission,copyright,creativecommons]{eptcs}
\providecommand{\event}{TARK 2023} % Name of the event you are submitting to

\usepackage{iftex}

\usepackage{latexsym}
\usepackage{amssymb}
\usepackage{amsmath}
\usepackage[all, knot]{xy}
\usepackage{graphicx}
\usepackage{hyperref}
\usepackage{multirow}

\newcommand{\M}{\mathcal{M}}
\newcommand{\N}{\mathcal{N}}

\newcommand{\B}{\mathsf{[B]}}

\newcommand{\lr}[1]{\langle #1 \rangle}

\newcommand{\LL}{\mathcal{L}}
\newcommand{\K}{\mathsf{[K]}}

\newcommand{\TBM}{[\mathsf{tB}^\mathtt{MS}]}
\newcommand{\TBMF}{[\mathsf{tB}^\mathtt{MS}_\mathtt{FS}]}
\newcommand{\KM}{[\mathsf{K}^\mathtt{MS}]}
\newcommand{\KMF}{[\mathsf{K}^\mathtt{MS}_\mathtt{FS}]}
\newcommand{\DK}{\langle \mathsf{K} \rangle}
\newcommand{\DB}{\langle \mathsf{B} \rangle}
\newcommand{\KWH}{[\mathsf{K}_\mathtt{wh}]}

\newcommand{\RB}{R_\mathsf{B}}

\newtheorem{theorem}{Theorem}[section]
\newtheorem{remark}[theorem]{Remark}
\newtheorem{lemma}[theorem]{Lemma}
\newtheorem{definition}[theorem]{Definition}
\newtheorem{proposition}[theorem]{Proposition}
\newtheorem{example}[theorem]{Example}

\newenvironment{proof} {\textsc{Proof.}\quad} {\hfill $\Box$\\}

\ifpdf
  \usepackage{underscore}         % Only needed if you use pdflatex.
  \usepackage[T1]{fontenc}        % Recommended with pdflatex
\else
  \usepackage{breakurl}           % Not needed if you use pdflatex only.
\fi

\title{Knowledge-wh and False Belief Sensitivity:\\
A Logical Study 
(An Extended Abstract)}
\author{Yuanzhe Yang
\institute{Department of Philosophy and Religious Studies\\
Peking University\\
Beijing, China}
\email{1900014924@pku.edu.cn}
}
\def\titlerunning{Knowledge-wh and False Belief Sensitivity}
\def\authorrunning{Y. Yang}

\begin{document}
\maketitle

\begin{abstract}
In epistemic logic, a way to deal with knowledge-wh is to interpret them as a kind of mention-some knowledge (MS-knowledge).
But philosophers and linguists have challenged both the sufficiency and necessity of such an account:
some argue that knowledge-wh has, in addition to MS-knowledge, also a sensitivity to false belief (FS);
others argue that knowledge-wh might only imply mention-some true belief (MS-true belief).
In this paper, we offer a logical study for all these different accounts.
We apply the technique of bundled operators, and introduce four different bundled operators - 
$\TBM^x \phi := \exists x (\B \phi \wedge \phi)$,
$\TBMF^x \phi := \exists x (\B \phi \wedge \phi) \wedge \forall x (\B \phi \to \phi)$,
$\KM^x \phi := \exists x \K \phi$ and
$\KMF^x \phi := \exists x \K \phi \wedge \forall x (\B \phi \to \phi)$ -,
which characterize the notions of MS-true belief, MS-true belief with FS, MS-knowledge and MS-knowledge with FS respectively.
We axiomatize the four logics which take the above operators (as well as $\K$) as primitive modalities on the class of $S4.2$-constant-domain models,
and compare the patterns of reasoning in the obtained logics, in order to show how the four accounts of knowledge-wh differ from each other, as well as what they have in common.
\end{abstract}

\section{Introduction}
In standard epistemic logic, for the most time, we deal with \emph{propositional knowledge} (or knowledge-\emph{that}):
that is, an agent knows \emph{that} $\phi$, where $\phi$ is a certain proposition.
However, this clearly does not exhaust our daily use of the notion of ``knowledge''.
Besides knowledge-that, we also frequently talk about various kinds of \emph{knowledge-wh}:
for example, I know \emph{how} to ride a bike, I know \emph{who} proved the incompleteness theorems, 
I know \emph{when} a certain meeting is held, I know \emph{where} to buy a certain book, I know \emph{what} is the password of my computer, I know \emph{why} a certain event happens, etc.

Thus, besides standard propositional knowledge, knowledge-wh also seems to be an interesting subject for epistemic logic to study.
There are already a number of logical studies of various kinds of knowledge-wh (e.g.\ know whether in \cite{DBLP:journals/rsl/FanWD15}, know why in \cite{xu2021logic}, know how in \cite{wang2015logic},\cite{DBLP:conf/ijcai/FervariHLW17}, \cite{DBLP:conf/icla/LiW17}, \cite{Naumov2018-NAUTWK} and \cite{DBLP:journals/synthese/Wang18}, just to name a few),
and a more general framework for logics of knowledge-wh is also proposed in \cite{Wang2017-WANANF}.
In this paper, following \cite{Wang2017-WANANF}, we will also focus mainly on the general logical structures shared by various kinds of knowledge-wh.

As suggested in \cite{Wang2017-WANANF} (following the philosophical stance of the so-called ``intellectualism'' initiated in \cite{Stanley2001-WILKHV}), in many cases, knowledge-wh can be interpreted as a kind of \emph{mention-some knowledge} (MS-knowledge for short):
for example, I know how to prove a theorem, iff \emph{there exists} some proof such that I know that this proof is a proof for the theorem;
I know where to buy newspapers, iff \emph{there exists} some place where I know I can buy newspapers, etc.
Then, in such cases, it seems that the logical structure of knowledge-wh can be formally captured by the first-order modal formula $\exists x \K \phi(x)$.\footnote{
However, as it is also noted in \cite{Wang2017-WANANF}, in some other situations, it seems more natural to interpret knowledge-wh in terms of mention-\emph{all}, rather than mention-\emph{some}, knowledge.
For example, when I say ``I know who came to the meeting yesterday'', it may mean that I know \emph{all} the people who came to that meeting, which should probably be formalized as, for example, $\forall x (\phi(x) \to \K \phi(x))$ or $\forall x (\K \phi(x) \vee \K \neg \phi(x))$. 
We will not deal with the mention-all reading of knowledge-wh in this paper, since the behavior of mention-all knowledge is rather different from mention-some knowledge, and it thus seems better to study it independently elsewhere.

In fact, axiomatization of mention-all knowledge in terms of $\forall x (\K \phi(x) \vee \K \neg \phi(x))$ has been studied in \cite{Zhougrad}, an unpublished undergraduate thesis.
}

However, while it is quite clear that in many situations, knowledge-wh does involve some kind of mention-some structure, it is not as clear whether MS-knowledge really is the right account for knowledge-wh in these situations.
In fact, both the sufficiency and necessity of such an account are challenged.

For example,
as it is argued in \cite{George2013-GEOKMR}, \cite{Phillips2018-PHIKWA}, \cite{Harris2019-HARKAF} and \cite{xiang2016complete},
knowledge-wh may not only involve mention-some knowledge, but also involve \emph{false belief sensitivity} (FS for short).
Let's consider the following scenario, adapted from one offered in \cite{George2013-GEOKMR}, to illustrate this point.
\begin{example}\label{news}
There are two stores, Newstopia and Paperworld.
Newstopia sells newspapers, while Paperworld sells only stationery.
Now, Alice \emph{knows} that Newstopia sells newspapers,
but also \emph{believes} erroneously that Paperworld sells newspapers.
\end{example}
In such a scenario, it is natural to judge that that Alice does not know where to buy newspapers (psychological experiments conducted in \cite{Phillips2018-PHIKWA} also show that such an intuition is shared by many people):
even though she has a MS-knowledge concerning where to buy newspapers,
it seems that her false belief that Paperworld sells newspapers would corrupt her knowledge-where.

Hence, maybe knowledge-wh should be sensitive to false belief:
that is, even under an MS-reading, maybe MS-knowledge should not be characterized by $\exists x \K \phi(x)$ alone, 
but should rather be characterized by $\exists x \K \phi(x) \wedge \forall x (\B \phi(x) \to \phi(x))$.

On the other hand, the necessity of the MS-knowledge account is also doubted.
For example, as it is argued in \cite{Carter2015-CARKAE-10}, it seems that knowledge-wh is subject to a kind of \emph{epistemic luck} which is not consistent with propositional knowledge.
Let's consider the following scenario, adapted from one offered in \cite{Carter2015-CARKAE-10}, to illustrate this point.
\begin{example}
Suppose that Bob believes that $w$ is a way to change light bulbs, and $w$ is indeed a reliable way to do so.
His belief is obtained by reading an instruction in a book.
However, unknown to him, all other contents in the book are erroneous, and it is merely due to a very rare print error that the instruction he read is correct.
\end{example}
In this case, Bob's true belief that $w$ is a way to change light bulbs is too lucky to be counted as his \emph{knowledge};
but nevertheless, it still seems natural to judge that Bob knows how to change light bulbs.

Then, it seems that sometimes a mention-some true belief (MS-true belief for short), i.e. $\exists x (\B \phi(x) \wedge \phi(x))$, is enough for knowledge-wh.
(In philosophical discussions, such a stance is sometimes called ``revisionary intellectualism'', 
which is first proposed in \cite{Cath2015-CATRIA}, in contrast to intellectualism.)

Of course, none of the arguments presented above is decisive.
But they do reveal an enormous complexity in the question concerning the nature of knowledge-wh.
Hitherto, no consensus concerning this question is reached in philosophical discussions,
and nor will we offer a determinate answer here.
On the contrary, in this paper, we will study \emph{all} the accounts mentioned above in a formal way.

In order to do so, we apply the technique of ``bundled operators''\footnote{
For a detailed introduction of such an idea, see \cite{Wang2017-WANANF} and \cite{Wang2018-WANBKT-2}.
}.
The general idea is that we pack a complex first-order modal formula (e.g.\ $\exists x \K \phi(x) \wedge \forall x (\B \phi(x) \to \phi(x))$) into the semantics of a single operator,
and study the logic which takes such an ``bundled operator'' as primitive modality.
By working in such languages with limited expressivity, we can focus on the behavior of the epistemic notion in which we are really interested, 
without being distracted by irrelevant notions which can also be expressed in a stronger language.
Moreover, with the help of bundled operators, we can study the complex notions in a compact manner.

In this paper, then, we will study the following four different bundled operators:\footnote{
The bundled operator $\KM^x$ is first introduced in \cite{Wang2017-WANANF} (the notation used there is $\Box^x$, though);
later, further study concerning its decidability and complexity is presented in \cite{DBLP:conf/fsttcs/PadmanabhaRW18}, and axiomatization in \cite{Wang2021-WANCTF-3}.
The result presented in this paper concerning this operator (namely, the axiomatization on $S4.2$), however, is new.

On the other hand, $\TBM$, $\TBMF$ and $\KMF$ are all novel bundled operators that have not yet been studied in literature.
}
\begin{center}
\begin{tabular}{lcl}
$\TBM^x \phi(x)$ & $:=$ & $\exists x (\B \phi(x) \wedge \phi(x))$\\
$\TBMF^x \phi(x)$ & $:=$ & $\exists x (\B \phi(x) \wedge \phi(x)) \wedge \forall x (\B \phi(x) \to \phi(x))$\\
$\KM^x \phi(x)$ & $:=$ & $\exists x \K \phi(x)$\\
$\KMF^x \phi(x)$ & $:=$ & $\exists x \K \phi(x) \wedge \forall x (\B \phi(x) \to \phi(x))$.
\end{tabular}
\end{center}

We will axiomatize the logics which take these operators plus an operator for propositional knowledge as primitive modalities on the class of $S4.2$-models, a class of models which characterizes knowledge, belief and their interactions in a reasonable way.
Completeness results will also be presented.
Moreover, we will compare the obtained logics, 
in order to show the differences and commonalities in the ways we reason about knowledge-wh, propositional knowledge and belief, which are logically implied by the different accounts of knowledge-wh.

\section{First-order $S4.2$-models}
First, we introduce the models we use to characterize knowledge and belief on the semantic level.

Since first-order quantifiers are involved in the notions of MS-knowledge, MS-true belief and FS,
we will use \emph{first-order Kripke models} as the semantic basis.
We fix a set of predicates $\mathcal{P}$.
A first-order Kripke model, then, is defined as follow:\footnote{
In this paper, we will not introduce function symbols and constants to our language.
Hence, we will also not consider functions and constants in the following definition.
}

\begin{definition}
A \emph{first-order Kripke model} is a $4$-tuple $\M = (W, D, R, \rho)$, where
\begin{itemize}
\item $W \neq \emptyset$ is the set of epistemically possible worlds of the model;
\item $D \neq \emptyset$ is the domain of the model;
\item $R \subseteq W^2$ is the accessibility relation among the possible worlds, which characterizes epistemic indistinguishability;
\item $\rho: \mathcal{P} \times W \to \wp (D^{<\omega})$ assigns each $n$-ary predicate an $n$-ary relation on each possible world.
\end{itemize}
(We may abbreviate the term ``first-order Kripke model'' simply as ``model'' in the following discussions.)
\end{definition}

Note that such a model can be interpreted rather freely on the conceptual level, so that it can characterize various kinds of knowledge-wh.
For example, if we want to characterize the knowledge-how of an agent, then the elements in $D$ can be interpreted as different methods or devices available for the agent in question,
and a predicate $P \in \mathcal{P}$ can be interpreted as a certain goal, while $a \in \rho(P, w)$ reads ``at the epistemically possible world $w$, $a$ is a way to achieve $P$''.
Similarly, if we want to characterize knowledge-where, then the elements in $D$ can be interpreted as different locations accessible for the agent,
while predicates in $\mathcal{P}$ are interpreted as properties of these locations.
Of course, in a similar fashion, different models can also be used to characterize knowledge-who, knowledge-when or knowledge-what.

Also note that we only consider \emph{constant-domain} models here:
all possible worlds in a model share the same domain.
This is mainly in order to avoid technical and conceptual complexities,
and we believe this is indeed a reasonable (though inevitably idealized) assumption.

Of course, since we use first-order Kripke models to characterize the epistemic states of an agent, the Kripkean part of such models should also possess certain frame properties.

It is a popular choice to use $S5$-models to characterize an agent's knowledge, but we will not use such models in this paper.
This is mainly because we need to deal with both knowledge and belief, as well as the interactions between them (moreover, in our discussion, the notion of belief should be interpreted in a rather strong sense, so we would prefer interaction principles like $\B \phi \to \B \K \phi$ to hold),
and we must also allow the possibility for false belief, in order for the notion of FS to make any sense at all.
This, however, seems to be a difficult task when knowledge is characterized by $S5$-models.

Hence, we will use \emph{$S4.2$-models} instead - 
that is, models which are reflexive, transitive and strongly convergent.\footnote{
Here, we use the notion of \emph{strong} convergence to define $S4.2$-models;
but elsewhere, when defining $S4.2$-models,
the notion of \emph{weak} convergence might be used instead (A frame $(W, R)$ is \emph{weakly} convergent, iff for all $w, v, v' \in W$ s.t. $wRv$ and $wRv'$, there is some $u \in W$ s.t. $vRu$ and $v'Ru$).
Standard modal logic cannot distinguish these two kinds of models (as noted in \cite{Stalnaker2006-STAOLO-2}), but some of the languages studied in this paper are strong enough to distinguish them.

We choose the stronger notion of convergence here, because it seems more favorable both technically and conceptually.
This is also in accordance with Stalnaker's note in \cite{Stalnaker2006-STAOLO-2}.
}
The formal definition is as follow:

\begin{definition}
A frame $(W, R)$ is \emph{strongly convergent}, iff for all $w \in W$, there is some $u \in W$ s.t. for all $v \in W$, if $wRv$, then $vRu$.

A model based on a reflexive, transitive and strongly convergent frame is called an $S4.2$-model.
\end{definition}

We find such models attractive, because the class of $S4.2$-models validates the logic of knowledge $\mathbf{S4.2}$, in which belief can be reasonably \emph{defined} in terms of knowledge by the definition $\B \phi := \DK \K \phi$ (as explained in \cite{Lenzen1978-LENRWI}, the underlying idea is that, if one knows that she does not know something, then she would not believe it; and if she does not believe something, then she would know by introspection that she does not know it).
Moreover, the logic for the belief defined in this way is $\mathbf{KD45}$, and we also have many intuitive interaction principles between knowledge and belief (e.g.\ $\K \phi \to \B \phi$, $\B \phi \to \K \B \phi$, $\neg \B \phi \to \K \neg \B \phi$, $\B \phi \to \B \K \phi$).
(It is Lenzen who first proposed $\mathbf{S4.2}$ as a logic for knowledge in \cite{Lenzen1978-LENRWI} and \cite{Lenzen1979-LENEBZ-2}, from a syntatic perspective.
Later, Stalnaker also studied $\mathbf{S4.2}$ from a more semantic perspective in \cite{Stalnaker2006-STAOLO-2}.)

Moreover, as it is noted by Stalnaker in \cite{Stalnaker2006-STAOLO-2}, in an $S4.2$-frame $(W, R)$, we can define the following relation $\RB$, which corresponds to the notion of belief defined in terms of knowledge:
\begin{definition}
Given a frame $(W, R)$,
$\RB \subseteq W^2$ is the relation which satisfies that for all $w, u \in W$, $w \RB u$ iff for all $v \in W$ s.t. $wRv$, $vRu$.
\end{definition}

It is not hard to check that if $(W, R)$ is an $S4.2$-frame, then $(W, \RB)$ is $KD45$.
Moreover, after we formally introduce the languages and their semantics,
it will be easy to check that $\RB$ corresponds to $\B$
in exactly the way $R$ corresponds to $\K$.

\section{Languages and semantics}
Now, we introduce the languages for the bundled operators, as well as their exact semantics.

We first fix a set of variables $\mathbf{X}$.
Then, for any $\KWH \in \{\TBM, \TBMF, \KM, \KMF\}$,
the corresponding language $\LL(\KWH)$ (and also $\LL_\approx(\KWH)$) is defined as follow:

\begin{definition}
$\LL(\KWH)$-formulas are defined recursively as follow:
$$\phi ::= P(y_1, ..., y_n) \mid \neg \phi \mid \phi \wedge \phi \mid \K \phi \mid \KWH^x \phi$$
where $P \in \mathcal{P}$, $n \geq 0$ and $x, y_1, ..., y_n \in \mathbf{X}$.

$\neg \K \neg \phi$ is denoted as $\DK \phi$;
$\B \phi$ is an abbreviation for $\DK \K \phi$.

$\vee$, $\to$ and $\leftrightarrow$ are defined in the usual way.

Moreover, let $\LL_\approx(\KWH)$ be the language obtained by further adding an identity relation $\approx$ to $\LL(\KWH)$.\footnote{
In the following discussions, we will be working in the language $\LL(\KWH)$ when we do not specifically mention the language in which we are working.
We will make it clear whenever we switch our working language to $\LL_\approx(\KWH)$.
}
\end{definition}

Corresponding to our definition of the bundled operators,
we define the semantics for the above languages recursively as follow:

\begin{definition}
Given a model $\M = (W, D, R, \rho)$, a $w \in W$ and an assignment $\sigma$ from $\mathbf{X}$ to $D$:
\begin{center}
\begin{tabular}{|lcl|}
\hline
$\M,w,\sigma \vDash P(x_1, ..., x_n) $ &  $\iff$ & $  (\sigma(x_1), ..., \sigma(x_n)) \in \rho(P,w)$\\
$\M,w,\sigma \vDash x \approx y$ & $\iff$ & $\sigma(x) = \sigma(y)$\\
$\M,w,\sigma \vDash \neg \phi $ & $\iff$ & $ \M,w,\sigma \not \vDash \phi$\\
$\M,w,\sigma \vDash \phi \wedge \psi $ & $\iff$ & $ \M,w,\sigma \vDash \phi$ and $\M,w,\sigma \vDash \phi$\\
\vspace{0.5em}
$\M,w,\sigma \vDash \K \phi$ & $\iff$ & For all $v \in W$, if $wRv$, then $\M,v,\sigma \vDash \phi$\\
\hline
$\M,w,\sigma \vDash \TBM^x \phi$ & $\iff$ & There is some $a \in D$, s.t. $\M,w,\sigma[x \mapsto a] \vDash \B \phi \wedge \phi$\\
\hline
\multirow{2}{*}{$\M,w,\sigma \vDash \TBMF^x \phi$} & \multirow{2}{*}{$\iff$} & (i) There is some $a \in D$, s.t. $\M,w,\sigma[x \mapsto a] \vDash \B \phi \wedge \phi$\\
&& (ii) For all $b \in D$, $\M,w,\sigma[x \mapsto b] \vDash \B \phi \to \phi$\\
\hline
$\M,w,\sigma \vDash \KM^x \phi$ & $\iff$ & There is some $a \in D$, s.t. $\M,w,\sigma[x \mapsto a] \vDash \K \phi$\\
\hline
\multirow{2}{*}{$\M,w,\sigma \vDash \KMF^x \phi$} & \multirow{2}{*}{$\iff$} & (i) There is some $a \in D$, s.t. $\M,w,\sigma[x \mapsto a] \vDash \K \phi$\\
&& (ii) For all $b \in D$, $\M,w,\sigma[x \mapsto b] \vDash \B \phi \to \phi$\\
\hline
\end{tabular}
\end{center}
where $\sigma[x \mapsto a]$ is the assignment which maps $x$ to $a$, and agrees with $\sigma$ on any other point.
\end{definition}

Note that we need not introduce an independent operator for belief, since $\B \phi$ is already defined by $\DK \K \phi$ in the languages given above.
It is also not hard to check that on $S4.2$-models, the semantics for $\B \phi$ defined this way is indeed the following one:
\begin{center}
\begin{tabular}{|c|}
\hline
$\M,w,\sigma \vDash \B \phi \iff$ For all $v \in W$, if $w \RB v$, then $\M,v,\sigma \vDash \phi$ \\
\hline
\end{tabular}
\end{center}

\section{The logics}
Then, we introduce the four logics, corresponding to our four accounts of knowledge-wh respectively.
Their axiomatizations are all obtained in a similar fashion:
generally speaking, we start from a standard  $\mathbf{S4.2}$ system for $\K$,
and then add axioms and rules to describe the behaviors of the bundled operators.

Below is a list of schemas of axioms and rules that will be used to offer the axiomatizations (in which the operator $\KWH$ should be substituted by $\TBM$, $\TBMF$, $\KM$ or $\KMF$ in the corresponding logics):\footnote{
Note that when we use the notation $\phi[y/x]$ to denote the formula obtained by replacing every free occurrences of $x$ in $\phi$ with $y$, we also implicitly assume that $y$ is \emph{admissible} for $x$ in $\phi$:
that is, $y$ does not appear in the scope of any operator of the form $\KWH^y$ in $\phi$.
}
\begin{center}
\renewcommand\arraystretch{1.1}
\begin{tabular}{l|l||l|l}
\multicolumn{4}{c}{\textbf{Axioms}} \\
$\mathtt{TBtoK_{wh}}$ &  $(\B \phi \wedge \phi)[y/x] \to \KWH^x \phi$ &
$\mathtt{KtoK_{wh}}$ & $\K \phi[y/x] \to \KWH^x \phi$\\
\vspace{0.5em}
$\mathtt{K_{wh}toFS}$ & $\KWH^x \phi \to (\B \phi \to \phi) [y / x]$
& $\mathtt{BtoBK_{wh}}$ & $\B \phi[y/x] \to \B \KWH^x \phi$\\
\multicolumn{4}{c}{\textbf{Rules}}\\
\multirow{2}{*}{$\mathtt{K_{wh}toTB}$} & \multicolumn{3}{l}{\multirow{2}{*}{ $\displaystyle{\frac{\vdash \psi_0 \to \K(\psi_1 \to \cdots \K (\psi_n \to \neg (\B \phi \wedge \phi)) \cdots )}
{\vdash \psi_0 \to \K(\psi_1 \to \cdots \K (\psi_n \to \neg \KWH^x \phi) \cdots )}}$}}\\\\
\multirow{2}{*}{$\mathtt{K_{wh}toK}$} & \multicolumn{3}{l}{\multirow{2}{*}{ $\displaystyle{\frac{\vdash \psi_0 \to \K(\psi_1 \to \cdots \K (\psi_n \to \neg \K \phi) \cdots )}
{\vdash \psi_0 \to \K(\psi_1 \to \cdots \K (\psi_n \to \neg \KWH^x \phi) \cdots )}}$}}\\\\
\multirow{2}{*}{$\mathtt{FS\&BtoK_{wh}}$} & \multicolumn{3}{l}{\multirow{2}{*}{ $\displaystyle{\frac{\vdash \psi_0 \to \K (\psi_1 \to \cdots \K (\psi_n \to (\B \phi \to \phi)) \cdots )}
{\vdash \psi_0 \to \K (\psi_1 \to \cdots \K (\psi_n \to (\B \phi [y/x] \to \KWH^x \phi)) \cdots )}} \quad$}}\\\\
\multirow{2}{*}{$\mathtt{FS\&KtoK_{wh}}$} & \multicolumn{3}{l}{\multirow{2}{*}{ $\displaystyle{\frac{\vdash \psi_0 \to \K (\psi_1 \to \cdots \K (\psi_n \to (\B \phi \to \phi)) \cdots )}
{\vdash \psi_0 \to \K (\psi_1 \to \cdots \K (\psi_n \to (\K \phi [y/x] \to \KWH^x \phi)) \cdots )}}$}}\\\\
\multicolumn{4}{c}{(In all the rules above, $n$ can be any natural number, and we require that $x \notin \bigcup_{i \leq n} FV(\psi_i)$)}\\
\end{tabular}
\vspace{0.5em}
\end{center}

By using rules like $\mathtt{K_{wh}toTB}$ or $\mathtt{FS\&KtoK_{wh}}$, we have sacrificed some intuitiveness for technical reasons, but the underlying idea is straightforward:
for example, $\mathtt{K_{wh}toTB}$ essentially says $\KWH^x \phi \to \exists x (\B \phi \wedge \phi)$, and $\mathtt{FS\&KtoK_{wh}}$ says $\forall x (\B \phi \to \phi) \wedge \K \phi[y/x] \to \KWH^x \phi$, in languages where the existential and universal quantifiers are not available.

With the help of the above axioms and rules, then, we can give the following four logics:
\begin{center}
\begin{tabular}{l|l}
$\mathbf{S4.2}^{\TBM}$ & $\mathbf{S4.2}^\K \oplus \{\mathtt{TBtoK_{wh}}, \mathtt{K_{wh}toTB}\}$\\
$\mathbf{S4.2}^{\TBMF}$ & $\mathbf{S4.2}^\K \oplus \{\mathtt{K_{wh}toFS}, \mathtt{BtoBK_{wh}},  \mathtt{K_{wh}toTB}, \mathtt{FS\&BtoK_{wh}}\}$\\
$\mathbf{S4.2}^{\KM}$ & $\mathbf{S4.2}^\K \oplus \{\mathtt{KtoK_{wh}}, \mathtt{K_{wh}toK}\}$\\
$\mathbf{S4.2}^{\KMF}$ & $\mathbf{S4.2}^\K \oplus \{\mathtt{K_{wh}toFS}, \mathtt{BtoBK_{wh}},  \mathtt{K_{wh}toK}, \mathtt{FS\&KtoK_{wh}}\}$
\end{tabular}
\end{center}

Moreover, for any $\KWH \in \{\TBM, \TBMF, \KM, \KMF\}$, when we work in the language $\LL_\approx(\KWH)$, let $\mathbf{S4.2}_\approx^{\KWH}$ be the logic defined as follows:

\begin{center}
\begin{tabular}{l|l}
$\mathbf{S4.2}_\approx^{\KWH}$ & $\mathbf{S4.2}^{\KWH} \oplus \{x \approx x,\; x \approx y \to (\phi[x/z] \to \phi[y/z]),\; x \not \approx y \to \K (x \not \approx y)\}$
\end{tabular}
\end{center}

Note that all the logics given here are non-normal, since they are all non-aggregative:
that is, $\KWH^x \phi \wedge \KWH^x \psi \to \KWH^x(\phi \wedge \psi)$ is not an inner theorem of $\mathbf{S4.2}^{\KWH}$ (or $\mathbf{S4.2}_\approx^{\KWH}$) for any $\KWH \in \{\TBM, \TBMF,\\ \KM, \KMF\}$
(in fact, in all these logics, $\KWH^x Px \wedge \KWH^x \neg Px$ is consistent).
Moreover, some of the logics are even non-monotone, as we will see below.

Then, we show the completeness theorem for these logics.

Since we are now dealing with bundled operators with more complex structures, the strategy to prove completeness theorems for the case of $\KM$ in \cite{Wang2017-WANANF} and \cite{Wang2021-WANCTF-3} cannot be directly applied here (moreover, axiomatization of the logic of $\KM$ on $S4.2$ has also not yet been studied).
Hence, we will develop a new strategy to prove completeness theorems for all the above logics in a uniform way.

\begin{theorem}\label{completeness}
$\mathbf{S4.2}^{\KWH}$ (as well as $\mathbf{S4.2}_\approx^{\KWH}$) is sound and strongly complete w.r.t.\ the class of $S4.2$-constant-domain models,
where $\KWH \in \{\TBM, \TBMF, \KM, \KMF\}$.
\end{theorem}

\begin{proof}
We only sketch the general idea of the proof here.
A detailed proof for the case of $\mathbf{S4.2}^{\KMF}$ can be found in the appendix.

Generally, we use maximal consistent sets (MCS) of formulas which also contain certain witness formulas to construct the canonical model.
The main difficulty is to ensure at the same time that (i) every MCS in the model contains all the witness formulas we need, (ii) every formula of the form $\DK \phi$ in an MCS has some accessible MCS containing $\phi$ as its witness, and
(iii) the canonical model is an $S4.2$-constant-domain model.

In order to construct such a model, we use a step-by-step method.
We start from a consistent set $\Gamma_0$, 
and extend consistent sets to MCS, add new formula sets as witnesses for formulas of the form $\DK \phi$,
and add witness formulas to formula sets during the same process.
The key is to ensure, at each step of the construction, that every formula set except $\Gamma_0$ is \emph{finite},
and all the information contained in a set is recorded in its predecessor with a formula of the form $\DK \phi$.
This ensures that we can always add witness formulas to formula sets using rules like $\mathtt{K_{wh}toTB}$ and $\mathtt{FS\&KtoK_{wh}}$.

Then, after countably many steps, we obtain a model which satisfies both (i) and (ii),
and is also an $S4$-constant-domain model.
Finally, we add another set of MCSs to the model to make it strongly convergent, so that we can obtain an $S4.2$-model.
\end{proof}

\begin{remark}
The above logics also have some interesting technical aspects.

For example, it is shown in \cite{DBLP:conf/fsttcs/PadmanabhaRW18} that the language $\LL(\KM)$ cannot distinguish constant-domain and increasing-domain models in general.
However, when we confine the models to $S4.2$-ones, $\LL(\KM)$ \emph{can} distinguish constant-domain and increasing-domain models,
and consequently, $\mathbf{S4.2}^{\KM}$ is not sound w.r.t.\ the class of $S4.2$-increasing-domain models (e.g.\, $\DK \KM^x \phi \to \KM^x \DK \phi$ is an inner theorem of $\mathbf{S4.2}^{\KM}$, but is not valid on $S4.2$-increasing-domain models).
In fact, for all $\KWH \in \{\TBM, \TBMF, \KM,\ \\ \KMF\}$, $\mathbf{S4.2}^{\KWH}$ is not sound w.r.t.\ $S4.2$-increasing-domain models.

Another interesting fact is that $\mathbf{S4.2}^{\TBMF}$ and $\mathbf{S4.2}^{\KMF}$ are able to distinguish $S4.2$-models (defined in terms of \emph{strong} convergence) and models which are reflexive, transitive but only weakly convergent.
The axiom $\mathtt{BtoBK_{wh}}: \B \phi[y/x] \to \B \KWH^x \phi$ does the trick.
When $\KWH = \TBM$ or $\KM$, on the other hand, we also have $\B \phi[y/x] \to \B \KWH^x \phi$ as an inner theorem of $\mathbf{S4.2}^{\KWH}$,
but in this case, the formula does not have the power to distinguish strong and weak convergence,
and consequently, $\mathbf{S4.2}^{\TBM}$ and $\mathbf{S4.2}^{\KM}$ are also sound w.r.t.\ the class of reflexive, transitive and weakly convergent models.
\end{remark}

\section{Comparisons}
Now, we have the formal ground to compare the different accounts of knowledge-wh.

\subsection{Differences}\label{Differences}
An interesting difference among the different accounts of knowledge-wh concerns the ways these accounts interact with propositional knowledge.

For example, consider \emph{positive introspection}.
Since we take $\mathbf{S4.2}$ to be the underlying logic for propositional knowledge, which is stronger than $\mathbf{S4}$, it is clear that propositional knowledge satisfies positive introspection:
$\K \phi \to \K \K \phi$ is an inner theorem of $\mathbf{S4.2}^{\KWH}$ for any $\KWH \in \{\TBM, \TBMF, \KM, \KMF\}$.
However, does knowledge-wh also have positive introspection?
To put it more formally, is $\KWH^x \phi \to \K \KWH^x \phi$ an inner theorem of $\mathbf{S4.2}^{\KWH}$?
The answer is as follow:

\begin{proposition}\label{PI1}
$\mathbf{S4.2}^{\KM} \vdash \KM^x \phi \to \K \KM^x \phi$,
but $\mathbf{S4.2}^{\KWH} \not \vdash \KWH^x \phi \to \K \KWH^x \phi$ when $\KWH \in \{\TBM, \TBMF, \KMF\}$.
\end{proposition}

The underlying reason for the failure of positive introspection in $\mathbf{S4.2}^{\TBM}$, $\mathbf{S4.2}^{\TBMF}$ and $\mathbf{S4.2}^{\KMF}$ is similar.
Essentially, this is because these accounts may involve true beliefs (the MS-true belief in $\TBM^x \phi$ or $\TBMF^x \phi$, or a true belief required by the FS condition in $\TBMF^x \phi$ or $\KMF^x \phi$), but positive introspection requires \emph{knowledge} rather than mere true belief, while the latter in general does not imply the former in an $\mathbf{S4.2}$ system.

The following proposition helps us make this point clear on the formal level.
Note that in the formulation of (a part of) the following proposition, we will also need the identity relation $\approx$ and the logic $\mathbf{S4.2}_\approx^{\KWH}$ which involves the axioms for $\approx$.

\begin{proposition}\label{PI2}
We have the following identities between logics:
\begin{center}
\begin{tabular}{lcl}
$\mathbf{S4.2}^{\TBM} \oplus \TBM^x \phi \to \K \TBM^x \phi$ & $=$ & $\mathbf{S4.2}^{\TBM} \oplus \B \phi \wedge \phi \to \K \phi$\\
$\mathbf{S4.2}^{\TBMF} \oplus \TBMF^x \phi \to \K \TBMF^x \phi$ & $=$ & $\mathbf{S4.2}^{\TBMF} \oplus \B \phi \wedge \phi \to \K \phi$\\
$\mathbf{S4.2}_\approx^{\KMF} \oplus \KMF^x \phi \to \K \KMF^x \phi$ & $=$ & $\mathbf{S4.2}_\approx^{\KMF} \oplus x \not \approx y \to (\B \phi \wedge \phi \to \K \phi)$
\end{tabular}
\end{center}
\end{proposition}

In other words, under our $\mathbf{S4.2}$ setting for propositional knowledge, requiring $\TBM^x \phi$ and $\TBMF^x \phi$ to satisfy positive introspection is in effect the same as requiring true belief to imply knowledge.
The case for $\KMF^x \phi$, on the other hand, is a bit more complex:
when $\KMF^x \phi$ satisfies positive introspection, either true belief implies knowledge, or there is at most one element in the domain (in which case the notion of FS is clearly trivialized).

A similar phenomenon also appears in the case of the formula $\KWH^x \phi \to \KWH^x \K \phi$.
Intuitively, the formula says that knowledge-wh offers the agent a way to obtain propositional knowledge:
for example, if we interpret $\KWH$ in terms of knowledge-how, then the formula says that if an agent knows how to achieve $\phi$, then she also knows how to make herself know that $\phi$.
In fact, Proposition \ref{PI1} and \ref{PI2} still hold after we substitute every occurrences of $\K \KWH^x \phi$ in these propositions with $\KWH^x \K \phi$, 
since
 $\KWH^x \K \phi \leftrightarrow \K \KWH^x \phi$ is in fact an inner theorem in $\mathbf{S4.2}^{\KWH}$ for all $\KWH \in \{\TBM, \TBMF, \KM, \KMF\}$. 

A more interesting difference among the different accounts of knowledge-wh concerns the \emph{monotonicity} of knowledge-wh.
We say our notion of knowledge-wh is \emph{monotone} if the following rule is admissible in the corresponding logic:

$$\mathtt{MONO} \quad \displaystyle{\frac{\vdash \phi \to \psi}{\vdash \KWH^x \phi \to \KWH^x \psi}}$$

The rule says that if $\psi$ follows logically from $\phi$,
then if an agent has knowledge-wh of $\phi$, then she automatically also has knowledge-wh of $\psi$.

Note that in order for this to hold, we need to assume that the agent we consider is logically omniscient;
and we have indeed assumed so in our underlying logic for propositional logic, $\mathbf{S4.2}^\K$.
However, even such a logically omniscient agent still may \emph{not} have a monotone notion of knowledge-wh, when FS is involved in our account of knowledge-wh.

The propositions below show how FS influences the monotonicity of knowledge-wh.
(Note that we need the identity relation $\approx$ to formulate Proposition \ref{MONO2}.)

\begin{proposition}\label{MONO1}
$\mathtt{MONO}$ is admissible in $\mathbf{S4.2}^{\TBM}$ and $\mathbf{S4.2}^{\KM}$,
but inadmissible in $\mathbf{S4.2}^{\TBMF}$ and $\mathbf{S4.2}^{\KMF}$.
\end{proposition}

\begin{proposition}\label{MONO2}
The following equivalences holds:
\begin{center}
\begin{tabular}{lcl}
$\mathbf{S4.2}_\approx^{\TBMF} \oplus \mathtt{MONO}$ & $=$ & $\mathbf{S4.2}_\approx^{\TBMF} \oplus x \not \approx y \to (\B \phi \to \phi)$\\
$\mathbf{S4.2}_\approx^{\KMF} \oplus \mathtt{MONO}$ & $=$ & $\mathbf{S4.2}_\approx^{\KMF} \oplus x \not \approx y \to (\B \phi \to \phi)$
\end{tabular}
\end{center}
\end{proposition}

As we can see, FS corrupts the monotonicity of knowledge-wh.
In fact, as it is shown in Proposition \ref{MONO2}, if we force $\TBMF^x \phi$ and $\KMF^x \phi$ to be monotone, then either the agent can have no false belief at all, or there is only one element in the domain of the model which characterizes her knowledge and belief - 
in both cases, the notion of FS is completely trivialized.
In this sense, we may say that FS is incompatible with the monotonicity of knowledge-wh in quite an essential way:
in order to retain the monotonicity of knowledge-wh, we have to give up FS completely.

\subsection{Commonalities}
As we have seen, different accounts of knowledge-wh behave rather differently when interacting with propositional knowledge.
However, when interacting with \emph{belief},
their behaviors are much more similar.

For example, the following proposition shows some inner theorems shared by all the logics presented above:\footnote{
Note that (iii) in the proposition below is in fact an \emph{axiom} in $\mathbf{S4.2}^{\TBMF}$ and $\mathbf{S4.2}^{\KMF}$.
}

\begin{proposition}
For all $\KWH \in \{\TBM, \TBMF, \KM, \KMF\}$,
the following are $\mathbf{S4.2}^{\KWH}$-theorems:
\begin{center}
\begin{tabular}{llll}
(i) & $\B \phi[y/x] \to \KWH^x \B \phi$ & (ii) & $\neg \B \phi[y/x] \to \KWH^x \neg \B \phi$\\
(iii) & $\B \phi[y/x] \to \B \KWH^x \phi$ & 
(iv) & $\B \KWH^x \phi \vee \B \neg \KWH^x \phi$
\end{tabular}
\end{center}
\end{proposition}

If we interpret $\KWH^x \phi$ in terms of knowledge-how, then (i) and (ii) say that if an agent believes / does not believe that some certain $y$ is a way to achieve $\phi$, then she knows how to make herself believe / not believe that $\phi$;
(iii) says that if the agent believes that some $y$ is a way to achieve $\phi$, then she also believes that she knows how to achieve $\phi$;
and (iv) says that an agent is ``confident'' concerning her own epistemic state:
for any $\phi$, she either believes that she knows how to $\phi$, or believes that she does not knows how to $\phi$. Note that in $\mathbf{S4.2}^{\K}$, we also have the interaction principles $\B \phi \to \K \B \phi$, $\neg \B \phi \to \K \neg \B \phi$ and $\B \phi \to \B \K \phi$ and $\B \K \phi \vee \B \neg \K \phi$ between propositional knowledge and belief;
hence, we may say that when interacting with propositional belief (rather than knowledge), 
our accounts of knowledge-wh show more aspects that resemble propositional knowledge.

Also note that from (i) and (iii),
we can deduce the following two formulas, respectively: 
\begin{center}
\begin{tabular}{llll}
(v) & $\KWH^x \phi \to \KWH^x \B \phi$ & (vi) &  $\KWH^x \phi \to \B \KWH^x \phi$
\end{tabular}
\end{center}
As we can see, though $\KWH^x \phi \to \KWH^x \K \phi$ and  $\KWH^x \phi \to \K \KWH^x \phi$ cannot be deduced in $\mathbf{S4.2}^{\KWH}$ when $\KWH \in \{\TBM, \TBMF, \KMF\}$, when the operator $\K$ is relaxed to $\B$, we obtain (v) and (vi), which are inner theorems of 
$\mathbf{S4.2}^{\KWH}$ for all $\KWH \in \{\TBM, \TBMF, \KM, \KMF\}$.

Another interesting commonality shared by all our logics (which also has to do with the interaction between knowledge-wh and belief) concerns what logic of knowledge-wh our agent believes.

In section \ref{Differences},
we have already shown some complexities in the reasoning about knowledge-wh:
for example, concerning positive introspection and monotonicity, different accounts yield different behaviors of knowledge-wh.
These complexities, however, only appear when \emph{we} reason about the knowledge-wh of an agent from an external perspective;
when \emph{the agent herself} reasons about her own knowledge-wh from within, all such complexities evaporate.

To put this point more rigidly, we introduce the following notion:

\begin{definition}
For any logic $\mathbf{L}$, let $\mathbf{L}_{\mathsf{B}} = \{\phi \mid \B \phi \in \mathbf{L}\}$.
\end{definition}

Intuitively, for a logic $\mathbf{L}$, $\mathbf{L}_\mathsf{B}$ collects all the formulas which $\mathbf{L}$ says that an agent believes.
In this sense, if $\mathbf{L}$ characterizes the epistemic states of an agent,
then $\mathbf{L}_\mathsf{B}$ characterizes the epistemic logic believed by this agent.

Then, with the help of this new notation, 
we can formulate the following theorem:

\begin{proposition}\label{S5inbelief}
For all $\KWH \in \{\TBM, \TBMF, \KM, \KMF\}$,
$\mathbf{S4.2}^{\KWH}_\mathsf{B}$ can be axiomatized by the following system:

\begin{center}
\begin{tabular}{l|l}
$\mathbf{S5}^\K$ & All axioms and rules of an $\mathbf{S5}$ system for $\K$\\
$\mathtt{KtoK_{wh}}$ & $\K \phi[y/x] \to \KWH^x \phi$\\
\multirow{2}{*}{$\mathtt{K_{wh}toK}^0$} & \multirow{2}{*}{$\displaystyle{\frac{\vdash \K \phi \to \psi}{\vdash \K^x \phi \to \psi}}$ (where $x \notin FV(\psi)$)}\\\\
\end{tabular}
\end{center}
\end{proposition}

It is also not hard to check that this system is equivalent to the system $\mathsf{SMLMSK}$ presented in \cite{Wang2017-WANANF}, a system in the language $\LL(\KM)$ which is sound and strongly complete w.r.t.\ the class of $S5$-models.

Hence, conceptually, the above theorem says that no matter which account of knowledge-wh we choose, it makes no difference for our agent:
the agent always believes that her knowledge-wh behaves in exactly the same way as MS-knowledge, and the logic for the underlying propositional knowledge is as strong as $\mathbf{S5}$.
In such a logic, of course knowledge-wh is monotone and satisfies positive introspection;
moreover, it even satisfies \emph{negative introspection}: $\neg \KWH^x \phi \to \K \neg \KWH^x \phi$ is in $\mathbf{S4.2}^{\KWH}_\mathsf{B}$ for all $\KWH \in \{\TBM, \TBMF, \KM, \KMF\}$.
On the other hand, all the subtle differences among the different accounts of knowledge-wh, generated from the gap between mere true belief and knowledge, as well as the peculiar behavior of the FS condition, are all invisible for the agent in question.

\section{Conclusion}
In this paper, we studied four bundled operators: $\TBM$, $\TBMF$, $\KM$ and $\KMF$, which correspond to the four different accounts of knowledge-wh.
We axiomatized the logics which take them (as well as $\K$) as primitive modalities on the class of $S4.2$-constant-domain models, and compared the ways we reason about knowledge-wh in different logics.

There many potential future works that can be done based on our work.

For example, we can further study the four bundled operators introduced in this paper.
We have only studied their behavior on $S4.2$-models, which characterize knowledge and belief in a highly idealized way;
our study of the obtained logics is also far from complete.
Hence, it seems interesting to study the logics obtained in this paper in greater detail, 
or to study the behavior of the bundled operators on other reasonable models for knowledge and belief (of course, we need not confine ourselves to Kripke models).
This may offer us a deeper understanding of the different accounts of knowledge-wh,
and may eventually help us decide which account is indeed the right one.

Moreover, the kind of step-by-step proof method applied in this paper can be generalized to study other complex epistemic notion.
For example, there are cases where it is better to understand knowledge-wh in terms of mention-all knowledge,
and there are also various competing accounts of these kinds of knowledge-wh, e.g.\
the \emph{weakly exhaustive} reading (first proposed in \cite{Karttunen1977-KARSAS}), the \emph{strongly exhaustive} reading (first proposed in \cite{groenendijk1982semantic}), and the \emph{intermediately exhaustive} reading (first raised, but soon rejected, in \cite{groenendijk1982semantic}, and later proposed again in \cite{spector2005exhaustive}), which can be formalized as $\forall x (\phi(x) \to \K \phi(x))$, $\forall x (\K \phi(x) \vee \K \neg \phi(x))$ and $\forall x (\phi(x) \to \K \phi(x)) \wedge \forall x (\B \phi(x) \to \phi(x))$, respectively.
Using the technique developed in this paper, we can easily pack these complex notions into bundled operators, and study their behavior.

Speaking on a more general level, the step-by-step method used in this paper can at least be generalized to any logic equipped with a set of ordinary modal operators $\{\Box_a\}_{a \in \tau}$ plus a set of bundled operators of the form $\blacksquare^x \phi := \exists x \alpha [\phi / p] \wedge \forall x \beta [\phi/p]$,
where $\alpha$ and $\beta$ are propositional modal formulas containing only one propositional symbol $p$, boolean connectives and operators in $\{\Box_a\}_{a \in \tau}$.
Our trick works no matter how complicated the structures of $\alpha$ and $\beta$ are, so a great deal of complex first-order modal notions can be handled in this way.

\bibliographystyle{eptcs}
\bibliography{generic}

\section*{Acknowledgement}
I owe the very idea of studying knowledge-wh with FS using bundled operators to Yanjing Wang, and 
I would also like to thank him for his insightful advice on this paper.

I would also like to thank Yimei Xiang for introducing works concerning MS-knowledge and FS in linguistics to me,
and Xun Wang for discussing matters concerning the $\exists\Box$-bundled fragment with me.
I should also thank three anonymous reviewers from TARK 2023 for their helpful advice on the writing of this paper,
and for identifying the typos and grammar mistakes in the paper.

Finally, I would like to thank the support of NSSF grant 19BZX135.

\appendix
\section{Appendix}

In the appendix, we show how to prove theorem \ref{completeness}, Proposition \ref{PI2} and Proposition \ref{MONO2}.

First, we consider theorem \ref{completeness}.
We only prove the case for $\mathbf{S4.2}^{\KMF}$, since the other cases can be proved in a similar way.
Moreover, for most of the time, we will be working in the language $\LL(\KMF)$, since our proof can easily be generalized to the case of $\LL_\approx(\KMF)$ with the help of some slight modifications.
We will demonstrate how to do so along the proof.

First, we check that the soundness result holds.

\begin{proposition}
$\mathbf{S4.2}^{\KMF}$ is sound w.r.t.\ the class of $S4.2$-constant-domain models.
\end{proposition}

\begin{proof}
We only prove that $\mathtt{BtoBK_{wh}}$ is valid on the class of $S4.2$-constant-domain models,
and $\mathtt{FS\&KtoK_{wh}}$ preserves validity on such models.

For $\mathtt{BtoBK_{wh}}$:

Let $\M = (W, R, D, \rho)$ be a $S4.2$-model, let $w \in W$ be arbitrary, and let $\sigma$ be an arbitrary assignment.
Assume that $\M,w,\sigma \vDash \B \phi[y/x]$.
Then, for all $v \in W$ s.t. $w \RB v$, $\M,v,\sigma \vDash \phi[y/x]$.

Then, let $v \in W$ be arbitrary, and assume that $w \RB v$.

We first show that $\M,v,\sigma \vDash \K \phi[y/x]$.
This is clear, since for all $u \in W$ s.t. $vRu$, it is easy to check that $w \RB u$,
and thus $\M,u,\sigma \vDash \phi[y/x]$.

Then, we show that for all $a \in D$, $\M,v,\sigma[x \mapsto a] \vDash \B \phi \to \phi$.
This is also clear: since $(W, \RB)$ is $KD45$, $v$ is $\RB$-reflexive.

Hence, it is easy to see that $\M,v,\sigma \vDash \KMF^x \phi$,
and thus $\M,w,\sigma \vDash \B \KMF^x \phi$.

For $\mathtt{FS\&KtoK_{wh}}$:

Let $\M = (W, D, R, \rho)$ be an arbitrary $S4.2$-model, and assume that 
$\psi_0 \to \K(\psi_1 \to \cdots \K(\psi_n \to (\B \phi \to \phi)) \cdots )$ is valid on $\M,w$, where $n$ is an arbitrary natural number;
and let $x$ be an variable s.t. $x \notin \bigcup_{i \leq n} FV(\psi_i)$.

Then, let $\sigma$ be an arbitrary assignment, and suppose (towards a contradiction) that $\M,w,\sigma \not \vDash \psi_0 \to \K(\psi_1 \to \cdots \K(\psi_n \to (\K \phi[y/x] \to \KMF^x \phi)) \cdots )$.
Then, there is some $w_0, w_1, ..., w_n \in W$,
s.t. $w = w_0 R w_1 R \cdots R w_n$,
$\M,w_i,\sigma \vDash \psi_i$ for all $i \leq n$,
and $\M,w_n,\sigma \not \vDash \K \phi[y/x] \to \KMF^x \phi$.
By the validity of $\psi_0 \to \K(\psi_1 \to \cdots \K(\psi_n \to (\B \phi \to \phi)) \cdots )$,
and since $x \notin \bigcup_{i \leq n} FV(\psi_i)$,
for all $a \in D$,
$\M,w_n,\sigma[x \mapsto a] \vDash \B \phi \to \phi$.
But then, since $\M,w_n, \sigma \vDash \K \phi[y/x]$,
$\M,w_n, \sigma[x \mapsto \sigma(y)] \vDash \K \phi$, and thus it should follow that $\M,w_n,\sigma \vDash \KMF^x \phi$,
causing a contradiction.
\end{proof}

It is also not hard to check that $\mathbf{S4.2}^{\KMF}$ has the following inner theorems, which will be used in our completeness proof.

\begin{tabular}{ll}
$\mathtt{NBK_{wh}toBNK_{wh}}$ & $\DB \KMF^x \phi \to \B \KMF^x \phi$\\
$\mathtt{BK_{wh}toK_{wh}B}$ & $\B \KMF^x \phi \to \KMF^x \B \phi$\\
$\mathtt{R}^{\KMF}$ & $\KMF^x \phi \leftrightarrow \KMF^y \phi[y/x]$ (where $y$ does not appear in $\phi$)\\
\end{tabular}

Now, we are ready to prove the completeness theorem.

As preparation, we first define the language $\LL^+(\KMF)$, which is obtained by adding countably many new variables to $\LL(\KMF)$.
We use $\mathbf{X}^+$ to denote the set of variables of $\LL^+(\KMF)$.

Then, we use a step-by-step method to prove the completeness theorem.
We first define the notion of a \emph{network}.
Note that when constructing such networks, the states will all be taken from a set of states $\{w_i \mid i \in \omega\}$,
which we fix in advance.

\begin{definition}
A \emph{network} is a triple $\N = (W, R, \nu)$, where
\begin{itemize}
\item $\{w_0\} \subseteq W \subseteq \{w_i \mid i \in \omega\}$;
\item $R \subseteq W^2$, and $(W, R)$ forms a tree where $w_0$ is the root;
\item $\nu$ assigns each element in $W$ a set of $\LL^+(\KMF)$-formulas.
\end{itemize}
\end{definition}

We also define the following two properties for the formula sets in a network:

\begin{definition}{($\mathtt{MS}$-property)}
An $\LL^+(\KMF)$-formula set $\Delta$ has \emph{$\mathtt{MS}$-property}, iff for all $\phi \in \LL^+(\KMF)$ and $x \in \mathbf{X}^+$, if $\KMF^x \phi \in \Delta$, then there is some $y \in \mathbf{X}^+$ s.t. $\K \phi [y/x] \in \Delta$.
\end{definition}

\begin{definition}{($\mathtt{FS}$-property)}
An $\LL^+(\KMF)$-formula set $\Delta$ has \emph{$\mathtt{FS}$-property}, iff for all $\phi \in \LL^+(\KMF)$ and $x, y \in \mathbf{X}^+$, if $\neg \KMF^x \phi, \K \phi[y/x] \in \Delta$, then there is some $z \in \mathbf{X}^+$ s.t. $(\B \phi \wedge \neg \phi)[z/x] \in \Delta$.
\end{definition}

Then, we define the notion of coherence and saturation for networks:

\begin{definition}
A network $\N = (W, R, \nu)$ is \emph{coherent}, iff the following conditions are satisfied:
\begin{itemize}
\item[(i)] $W$ is finite;
\item[(ii)] For all $w \in W$, $\nu(w)$ is $\mathbf{S4.2}^{\KMF}$-consistent; and for all $w \in W \setminus \{w_0\}$, $\nu(w)$ is finite;
\item[(iii)] For all $w, v \in W$ s.t. $wRv$, there is some $\psi$ s.t. $\vdash \psi \leftrightarrow \bigwedge \nu(v)$ and $\DK \psi \in \nu(w)$;
\item[(iv)] There are countably many variables in $\mathbf{X}^+$ which do not appear in $\nu(w)$ for any $w \in W$.
\end{itemize}
\end{definition}

\begin{definition}
A network $\N = (W, R, \nu)$ is \emph{saturated}, iff for all $w \in W$ and $\phi \in \LL^+(\KMF)$, the following holds:
\begin{itemize}
\item[(i)] $\nu(w)$ is a MCS of $\LL^+(\KMF)$-formulas;
\item[(ii)] If $\K \phi \in \nu(w)$, then for all $v \in W$ s.t. $wRv$, $\phi \in \nu(v)$;
\item[(iii)] If $\DK \phi \in \nu(w)$, then there is some $v \in W$ s.t. $wRv$ and $\phi \in \nu(v)$;
\item[(iv)] $\nu(w)$ has the $\mathtt{MS}$-property;
\item[(v)] $\nu(w)$ has the $\mathtt{FS}$-property.
\end{itemize}
\end{definition}

Then, corresponding to the requirements of saturation, we also introduce the following notion of \emph{defects}:

\begin{definition}
The possible kinds of \emph{defects} we may find on a state on a $w \in W$ in a network $\N = (W, R, \nu)$ are as follow:
\begin{itemize}
\item[(d1)] $\phi \notin \nu(w)$ and $\neg \phi \notin \nu(w)$
\item[(d2)] $\K \phi \in \nu(w)$, but there is some $v \in W$ s.t. $wRv$ and $\phi \notin \nu(v)$
\item[(d3)] $\DK \phi \in \nu(w)$, but there is no $v \in W$ s.t. $wRv$ and $\phi \in \nu(v)$
\item[(d4)] $\KMF^x \phi \in \nu(w)$, but there is no $y \in \mathbf{X}^+$ s.t. $\K \phi[y/x] \in \nu(w)$
\item[(d5)] $\neg \KMF^x \phi, \K \phi[y/x] \in \nu(w)$, but there is no $z \in \mathbf{X}^+$ s.t. $(\B \phi \wedge \neg \phi)[z/x] \in \nu(w)$
\end{itemize}
where $w \in \{w_i \mid i \in \omega\}$, $\phi \in \LL^+(\KMF)$ and $x \in \mathbf{X}^+$.
\end{definition}

Then, we prove the \emph{repair lemma},
which shows how to repair defects in a coherent network, while maintaining its coherence.

\begin{lemma}{(Repair lemma)}\label{repair}
For any coherent network $\N = (W, R, \nu)$ and any defect $(d)$ of $\N$,
then there is a coherent network $\N' = (W', R', \nu')$ s.t. $W \subseteq W'$, $R \subseteq R'$, 
$\nu(w) \subseteq \nu'(w)$ for all $w \in W$, and $\N'$ does not has $(d)$.
\end{lemma}

\begin{proof}
Let $\N = (W, R, \nu)$ be a coherent network, and assume that $\N$ has a defect $(d)$ for some $w_m \in W$ and $\phi \in \LL^+(\KMF)$.

Since $(W, R)$ forms a tree where $w_0$ is the root, there is a unique path $w_0 = v_0 R v_1 R \cdots R v_n = w_m$ in $\N$ for some $n \in \omega$.
Then, since $\N$ is coherent, for all $1 \leq i \leq n$, let $\psi_i$ stand for the formula s.t. $\vdash \psi_i \leftrightarrow \bigwedge \nu(v_i)$ and $\DK \psi_i \in \nu(v_{i-1})$.
Then, it is easy to see that 
$$\nu(v_0) \vdash \DK(\psi_1 \wedge \DK (\psi_2 \wedge \cdots \DK(\psi_{n-1} \wedge \DK \psi_n) \cdots))$$

We then consider five cases.\\

Case 1: $(d)$ is of the kind $(d1)$.
That is, $\phi \notin \nu(v_n)$ and $\neg \phi \notin \nu(v_n)$.
Then, it is easy to check that 
$$\nu(v_0) \vdash \DK(\psi_1 \wedge \DK (\psi_2 \wedge \cdots \DK(\psi_n \wedge \phi) \cdots )) \vee \DK(\psi_1 \wedge \DK (\psi_2 \wedge \cdots \DK(\psi_n \wedge \neg \phi) \cdots ))$$
Then, at least one of the disjuncts is consistent with $\nu(v_0) = \nu(w_0)$.
We only consider the case where the former is consistent with $\nu(w_0)$, since the other case is similar.
In this case, let 
\begin{align*}
\nu' = & \{(w, \nu(w)) \mid w \neq v_i \text{ for all } i \leq n\}\\
\cup &\{(v_n, \nu(v_n) \cup \{\phi\})\}\\
\cup &\{(v_i, \nu(v_i) \cup \{\DK(\psi_{i+1} \wedge \cdots \DK(\psi_n \wedge \phi) \cdots )\}) \mid i < n\}
\end{align*}
and let $\N' = \lr{W, R, \nu'}$.
It is then easy to check that $\N'$ is coherent, and does not have the defect $(d)$.\\

Case 2: $(d)$ is of the kind $(d2)$.
That is, $\K \phi \in \nu(v_n)$, but there is some $u \in W$ s.t. $v_nRu$ and $\phi \notin \nu(u)$.
Since $\N$ is coherent, there is some $\psi_u$ s.t. $\vdash \psi_u \leftrightarrow \bigwedge \nu(u)$ and $\DK \psi_u \in \nu(v_n)$.
Hence, it is easy to check that 
$$\nu(v_0) \vdash \DK(\psi_1 \wedge \DK(\psi_2 \wedge \cdots \DK(\psi_n \wedge \DK(\psi_u \wedge \phi)) \cdots ))$$
Then, let 
\begin{align*}
\nu' = & \{(w, \nu(w)) \mid w \neq v_i \text{ for all } i \leq n\}\\
\cup &\{(u, \nu(u) \cup \{\phi\})\}\\
\cup &\{(v_i, \nu(v_i) \cup \{\DK(\psi_{i+1} \wedge \cdots \DK(\psi_n \wedge \phi) \cdots )\}) \mid i \leq n\}
\end{align*}
It is easy to check that $\N' = (W, R, \nu')$ is still coherent, and does not have the defect $(d)$.\\

Case 3: $(d)$ is of the kind $(d3)$.
That is, $\DK \phi \in \nu(v_n)$, but there is no $u \in W$ s.t. $v_nRu$ and $\phi \in \nu(u)$.
Since $W$ is finite, there is some $\{w_i \mid i \in \omega\} \setminus W \neq \emptyset$.
Then, let $w_k$ be the element in $\{w_i \mid i \in \omega\} \setminus W$ with the least index number,
and let $W' = W \cup \{w_k\}$, $R' = R \cup \{(v_n, w_k)\}$, and $\nu' = \nu \cup \{(w_k, \{\phi\})\}$.
It is easy to check that $\N = (W', R', \nu')$ is coherent, but does not have $(d)$.\\

Case 4: $(d)$ is of the kind $(d4)$.
That is, $\KMF^x \phi \in \nu(v_n)$, but there is no $y \in \mathbf{X^+}$ s.t. $\K \phi[y/x] \in \nu(v_n)$.
Then, let $y \in \mathbf{X^+}$ be a variable that does not appear in $\nu(w)$ for any $w \in W$,
and suppose (towards a contradiction) that 
$$\nu(v_0) \vdash \K(\psi_1 \to \K(\psi_2 \to \cdots \K(\psi_n \to \neg \K \phi[y/x]) \cdots ))$$
Then, by $\mathtt{K_{wh}toK}$ (and $\mathtt{R}^{\KMF}$), we have
$$\nu(v_0) \vdash \K(\psi_1 \to \K(\psi_2 \to \cdots \K(\psi_n \to \neg \KMF^x \phi) \cdots ))$$
which contradicts the fact that $\nu(v_0) = \nu(w_0)$ is consistent.
Hence, $\DK(\psi_1 \wedge \DK(\psi_2 \wedge \cdots \DK(\psi_n \wedge \K \phi[y/x]) \cdots ))$ is consistent with $\nu(v_0) = \nu(w_0)$.
Hence, let 
\begin{align*}
\nu' = & \{(w, \nu(w)) \mid w \neq v_i \text{ for all } i \leq n\}\\
\cup &\{(v_n, \nu(v_n) \cup \{\K \phi[y/x]\})\}\\
\cup &\{(v_i, \nu(v_i) \cup \{\DK(\psi_{i+1} \wedge \cdots \DK(\psi_n \wedge \K \phi[y/x]) \cdots )\}) \mid i < n\}
\end{align*}
It is easy to check that $\N' = (W, R, \nu)$ is still coherent, and does not have the defect $(d)$.\\

Case 5: $(d)$ is of the kind $(d5)$.
That is, $\neg \KMF^x \phi \in \nu(v_n)$ and $\K \phi[y/x] \in \nu(v_n)$, but there is no $z \in \mathbf{X^+}$ s.t. $(\B \phi \wedge \neg \phi)[z/x] \in \nu(v_n)$.
Then, let $z \in \mathbf{X^+}$ be a variable that does not appear in $\nu(w)$ for any $w \in W$,
and suppose (towards a contradiction) that 
$$\nu(v_0) \vdash \K(\psi_1 \to \K(\psi_2 \to \cdots \K(\psi_n \to (\B \phi[z/x] \to \phi[z/x])) \cdots ))$$
Then, by $\mathtt{FS\&KtoK_{wh}}$ (and $\mathtt{R}^{\KMF}$), we have
$$\nu(v_0) \vdash \K(\psi_1 \to \K(\psi_2 \to \cdots \K(\psi_n \to (\K \phi[y/x] \to \KMF^x \phi)) \cdots ))$$
which contradicts the fact that $\nu(v_0) = \nu(w_0)$ is consistent.
Hence, $\DK(\psi_1 \wedge \DK(\psi_2 \wedge \cdots \DK(\psi_n \wedge (\B \phi \wedge \neg \phi)[z/x]) \cdots ))$ is consistent with $\nu(v_0) = \nu(w_0)$.
Hence, let 
\begin{align*}
\nu' = & \{(w, \nu(w)) \mid w \neq v_i \text{ for all } i \leq n\}\\
\cup &\{(v_n, \nu(v_n) \cup \{(\B \phi \wedge \neg \phi)[z/x]\})\}\\
\cup &\{(v_i, \nu(v_i) \cup \{\DK(\psi_{i+1} \wedge \cdots \DK(\psi_n \wedge (\B \phi \wedge \neg \phi)[z/x]) \cdots )\}) \mid i < n\}
\end{align*}
It is easy to check that $\N' = (W, R, \nu)$ is still coherent, and does not have the defect $(d)$.
\end{proof}

Then, we can easily show that every coherent network can be extended into a saturated network.

\begin{lemma}
For any coherent network $\N = (W, R, \nu)$,
there exists a saturated network $\N' = (W', R', \nu')$ s.t. $W \subseteq W'$, $R \subseteq R'$ and $\nu(w) \subseteq \nu'(w)$ for all $w \in W$.
\end{lemma}

\begin{proof}
Let $\N = (W, R, \nu)$ be a coherent network.

It is not hard to see that there are only countably many possible defects.
Hence, we can enumerate them as $(d)_1$, $(d)_2$, $(d)_3$, \dots

Then, we define a countable sequence of networks $\N_i = (W_i, R_i, \nu_i)$ ($i \in \omega$) recursively as follow:

\begin{itemize}
\item $\N_0 = \N$;
\item Given a coherent network $\N_k$, let $(d)_m$ be the defect of $\N_k$ with the least index number (note that according to our definition of coherence, $\N_k$ necessarily has defects), and let $\N_{k+1} = (W_{k+1}, R_{k+1}, \\  \nu_{k+1})$ be a coherent network which does not has $(d)_m$, and also satisfies that $W_k \subseteq W_{k+1}$, $R_k \subseteq R_{k+1}$, 
$\nu_k(w) \subseteq \nu_{k+1}(w)$ for all $w \in W_k$.
The existence of such a network is guaranteed by lemma \ref{repair}.
\end{itemize}

Then, let $\N' = (W', R', \nu')$, where
\begin{itemize}
\item $W' = \bigcup_{i \in \omega} W_i$;
\item $R' = \bigcup_{i \in \omega} R_i$;
\item For all $w \in W$, $\nu'(w) = \bigcup_{i \geq k} \nu_i(w)$, where $k$ is the least number s.t. $w \in W_k$.
\end{itemize}

It is not hard to see that $\N'$ is a saturated network s.t. $W_0 \subseteq W$, $R_0 \subseteq R$ and $\nu_0(w) \subseteq \nu(w)$ for all $w \in W_0$.
\end{proof}

Then, we show how to induce a canonical model from a saturated network.

\begin{definition}
Given a saturated network $\N = (W, R, \nu)$, $\M^c_\N = (W_\N^c, R_\N^c, D_\N^c, \rho_\N^c)$ is the model induced from $\N$, where
\begin{itemize}
\item $W_\N^c = \{\nu(w) \mid w \in W\} \cup FC$,\\
where $FC = \{\Theta \mid \Theta \text{ is a MCS in } \LL^+(\KMF), \{\phi \mid \B \phi \in \nu(w_0)\} \subseteq \Theta\}$;\footnote{
$FC$ stands for \emph{F}inal \emph{C}luster.
In fact, we can show that for all $\Delta \in W^c$ and $\Theta \in FC$, $\Delta R^c \Theta$, which justifies our naming.
}
\item $D_\N^c = \mathbf{X^+}$;
\item $R_\N^c$ satisfies that for all $\Delta, \Theta \in W^c$, $\Delta R^c \Theta$ iff for all $\phi \in \LL^+(\KMF)$, if $\K \phi \in \Delta$, then $\phi \in \Theta$;
\item $\rho_\N^c$ satisfies that for all $\Delta \in W_\N^c$, $\bar{x} \in (D^c_\N)^{<\omega}$ and $P \in \mathcal{P}$, $\bar{x} \in \rho^c(P, \Delta)$ iff $P\bar{x} \in \Delta$.
\end{itemize}
We may drop the subscript $\N$ when the context is clear.
\end{definition}

\begin{remark}
If we are working in the language $\LL_\approx(\KMF)$, then we let $D^c_\N = \{[x] \mid x \in \mathbf{X^+}\}$, where $[x] = \{y \in \mathbf{X^+} \mid x \approx y \in \nu(w_0)\}$.
\end{remark}

We then show that a model induced from a saturated network is indeed $S4.2$,
and also has all the properties we need.

\begin{lemma}
For any saturated network $\N$, $\M^c_\N$ satisfies the following:
\begin{itemize}
\item[(i)] $\M^c_\N$ is an $S4.2$-model;
\item[(ii)] For all $\Delta \in W^c$ and $\DK \phi \in \Delta$, there is some $\Delta' \in W^c$ s.t. $\Delta R^c \Delta'$ and $\phi \in \Delta'$;
\item[(iii)] For all $\Delta \in W^c$, $\Delta$ has the $\mathtt{MS}$-property and the $\mathtt{FS}$-property.
\end{itemize}
\end{lemma}

\begin{proof}
Let $\M^c_\N$ be an arbitrary model induced from a saturated network $\N = (W, R, \nu)$.

For item (i):
By the definition of $R^c$ and the canonicity of $\mathtt{T}^\K$ and $\mathtt{4}^\K$, it is easy to see that $\M^c_\N$ is reflexive and transitive.

We then show that $\M^c_\N$ is strongly convergent.

Clearly $FC \neq \emptyset$,
since $\DB \top \in \nu(w_0)$.

Then, we show that for all $\Delta \in W^c$ and $\Theta \in FC$, $\Delta R^c \Theta$.
Let $\Delta \in W^c$ and $\Theta \in FC$ be arbitrary.
We consider two cases:

Case 1: there is some $w \in W$ s.t. $\Delta = \nu(w)$.
Let $\K \phi \in \nu(w)$ be arbitrary.
It is easy to see that $\nu(w_0) R^c \nu(w)$;
hence, $\DK \K \phi \in \nu(w_0)$, i.e. $\B \phi \in \nu(w_0)$.
Hence, by definition, $\phi \in \Theta$.
Thus, $\nu(w) R^c \Theta$.

Case 2: $\Delta \in FC$.
Let $\K \phi \in \Delta$ be arbitrary.
Then, since $\Delta \in FC$, $\DB \K \phi \in \nu(w_0)$, i.e. $\K \DK \K \phi \in \nu(w_0)$.
Then, by $\mathtt{T}^\K$, $\DK \K \phi \in \nu(w_0)$, i.e. $\B \phi \in \nu(w_0)$.
Hence, $\phi \in \Theta$ and
thus, $\Delta R^c \Theta$.

Therefore, $\M_\N^c$ is strongly convergent.\\

For item (ii):
Since $\N$ is saturated, we only need prove that for all $\Theta \in FC$ and $\DK \phi \in \Theta$, there is some $\Theta' \in W^c$ s.t. $\Theta R^c \Theta'$ and $\phi \in \Theta'$.

Let $\Theta \in FC$, $\DK \phi \in \Theta$ be arbitrary.
Then, since $\Theta \in FC$, $\DB \DK \phi \in \nu(w_0)$, i.e. $\K \DK \DK \phi \in \nu(w_0)$.
Hence, by $\mathtt{4}^\K$, $\K \DK \phi \in \nu(w_0)$, i.e. $\DB \phi \in \nu(w_0)$,
and thus, there is some $\Theta' \in FC$ s.t. $\phi \in \Theta'$.
Then, as we have already proved, $\Theta R^c \Theta'$.\\

For item (iii):
Again, since $\N$ is saturated, we only need to prove that every $\Theta \in FC$ has the $\mathtt{MS}$-property and the $\mathtt{FS}$-property.

Let $\Theta \in FC$ and $\phi \in \LL^+(\KWH)$ be arbitrary.

First, assume that $\KWH^x \phi \in \Theta$.
Then, $\DB \KWH^x \phi \in \nu(w_0)$,
and thus, by $\mathtt{NBK_{wh}toBNK_{wh}}$, $\B \K^x \phi \in \nu(w_0)$.
Then, by $\mathtt{BK_{wh}toK_{wh}B}$, $\K^x \B \phi \in \nu(w_0)$.
Then, since $\N$ is saturated, $\nu(w_0)$ has the $\mathtt{MS}$-property, and thus there is some $y \in \mathbf{X^+}$ s.t. $\K \B \phi[y/x] \in \nu(w_0)$.
Hence, $\B \phi[y/x] \in \nu(w_0)$, and thus $\B \K \phi[y/x] \in \nu(w_0)$.
Hence, $\K \phi[y/x] \in \Theta$.

Next, assume that $\KWH^x \phi \in \Theta$.
Then, $\neg \B \KWH^x \phi \in \nu(w_0)$,
and thus for all $y \in \mathbf{X}^+$,
$\neg \B \phi[y/x] \in \nu(w_0)$ by $\mathtt{BtoBK_{wh}}$.
Hence, for all $y \in \mathbf{X}^+$, $\B \neg \B \phi[y/x] \in \nu(w_0)$,
and thus $\B \phi[y/x] \notin \Theta$.
\end{proof}

Then, it is routine to prove the truth lemma:
\begin{lemma}
For all $\M^c_\N$ induced from a saturated network $\N$,
for all $\Delta \in W^c$ and $\phi \in \LL^+(\KMF)$, $\M_\N^c,\Delta,\sigma^c \vDash \phi \iff \phi \in \Delta$,
where $\sigma^c$ is the assignment s.t. $\sigma^c(x) = x$ for all $x \in \mathbf{X}^+$.
\end{lemma}

\begin{remark}
If we are working in the language $\LL_\approx(\KMF)$, then in the formulation of the above lemma, we let $\sigma^c$ be the assignment s.t. $\sigma^c(x) = [x]$ for all $x \in \mathbf{X}^+$.
\end{remark}

Finally, notice that for any $\mathbf{S4.2}^{\KMF}$-consistent set $\Gamma$ of $\LL(\KMF)$-formulas $\Gamma_0$,
$(\{w_0\}, \emptyset, \{(w_0, \Gamma_0)\})$ is a coherent network.
Hence, it can be extended into a saturated network $\N'$, from which we can induce a canonical model $\M^c_{\N'}$, such that $\M^c_{\N'},\nu'(w_0),\sigma^c \vDash \Gamma_0$.

Hence, we have the following completeness theorem:

\begin{theorem}
$\mathbf{S4.2}^{\KMF}$ (as well as $\mathbf{S4.2}_\approx^{\KMF}$) is sound and strongly complete w.r.t.\ the class of $S4.2$-constant-domain models.
\end{theorem}

Then, we consider Proposition \ref{PI2}.
The cases for $\TBM^x \phi$ and $\TBMF^x \phi$ are relatively easy, since it is easy to see that $x \notin FV(\phi)$, $\TBM^x \phi$ and $\TBMF^x \phi$ are equivalent to $\B \phi \wedge \phi$.
Hence, we only prove the following proposition here:

\begin{proposition}
The following equivalence holds:
\begin{center}
\begin{tabular}{lcl}
$\mathbf{S4.2}_\approx^{\KMF} \oplus \KMF^x \phi \to \K \KMF^x \phi$ & $=$ & $\mathbf{S4.2}_\approx^{\KMF} \oplus x \not \approx y \to (\B \phi \wedge \phi \to \K \phi)$
\end{tabular}
\end{center}
\end{proposition}

\begin{proof}
We first show that $x \not \approx y \to (\B \phi \wedge \phi \to \K \phi)$ can be deduced in $\mathbf{S4.2}_\approx^{\KMF} \oplus \KMF^x \phi \to \K \KMF^x \phi$.
It is easy to check that $\vdash \phi \wedge x \not \approx y \to\KMF^z (x \approx z \to \phi[z/x])$, where $z$ is a fresh variable.
Then, by positive introspection, $\vdash \phi \wedge x \not \approx y \to \K \KMF^z (x \approx z \to \phi[z/x])$,
and by $\mathtt{K_{wh}toFS}$, 
$\vdash \phi \wedge x \not \approx y \to \K (\B (x \approx x \to \phi) \to (x \approx x \to \phi))$.
Hence, $\vdash \phi \wedge x \not \approx y \to \K (\B \phi \to \phi)$,
and thus $\vdash \phi \wedge x \not \approx y \to (\B \phi \to \K \phi)$.
Hence, $\vdash x \not \approx y \to (\B \phi \wedge \phi \to \K \phi)$.

Then, we show that $\KMF^x \phi \to \K \KMF^x \phi$ can be deduced in $\mathbf{S4.2}^{\KMF}_\approx \oplus x \not \approx y \to (\B \phi \wedge \phi \to \K \phi)$.
Equivalently, we show that $\KMF^x \phi \wedge \KMF^x \K \phi$ can be deduced.
Since $\vdash x \not \approx y \to (\B \phi \wedge \phi \to \K \phi)$ for some fresh $y$,
by $\mathtt{K_{wh}toK}$,
we have $\vdash \KWH^y(x \not \approx y) \to (\B \phi \wedge \phi \to \K \phi)$.
Then, we first show that $\vdash \KMF^x \phi \to (\B \phi \to \K \phi)$.
On the one hand, it is easy to check that we have $\vdash \neg \KMF^y (x \not \approx y) \to z \approx x$ (where $z$ is a fresh variable), and thus 
$\vdash \K \phi[z/x] \wedge \neg \KMF^y (x \not \approx y) \to (\B \phi \to \K \phi)$.
Hence, $\vdash \KMF^x \phi \wedge \neg \KMF^y (x \not \approx y) \to (\B \phi \to \K \phi)$ by $\mathtt{K_{wh}toK}$ (and $\mathtt{R}^{\KMF}$).
On the other hand, $\vdash \KMF^x \phi \wedge \KMF^y (x \not \approx y) \to (\B \phi \wedge \phi \to \K \phi)$,
and thus $\vdash \KMF^x \phi \wedge \KMF^y (x \not \approx y) \to (\B \phi \to \K \phi)$ by $\mathtt{K_{wh}toFS}$.
Hence, $\vdash \KMF^x \phi \to (\B \phi \to \K \phi)$.
Then, by $\mathtt{FS\&KtoK_{wh}}$ and $\mathtt{4}^\K$,
$\vdash \KMF^x \phi \to (\K \phi \to \KMF^x \K \phi)$,
and thus $\vdash \KMF^x \phi \to \KMF^x \K \phi$ by $\mathtt{K_{wh}toK}$.
\end{proof}

Finally, for Proposition \ref{MONO2},
we only prove the case for $\KMF$, since the case for $\TBMF$ is similar.
That is, we prove the following proposition:

\begin{proposition}
The following equivalence holds:
\begin{center}
\begin{tabular}{lcl}
$\mathbf{S4.2}_\approx^{\KMF} \oplus \mathtt{MONO}$ & $=$ & $\mathbf{S4.2}_\approx^{\KMF} \oplus x \not \approx y \to (\B \phi \to \phi)$
\end{tabular}
\end{center}
\end{proposition}

\begin{proof}
We first show that $x \not \approx y \to (\B \phi \to \phi)$ can be deduced in $\mathbf{S4.2}_\approx^{\KMF} \oplus \mathtt{MONO}$.
Clearly $\vdash (x \approx y) \wedge (x \not \approx y) \to \phi$, i.e. $\vdash x \approx y \to (x \not \approx y \to \phi)$.
Then, by $\mathtt{MONO}$, $\vdash\KMF^y(x \approx y) \to \KMF^y (x \not \approx y \to \phi)$.
It is also easy to check that $\vdash \KMF^y(x \approx y)$.
Hence, $\vdash \KMF^y (x \not \approx y \to \phi)$.
Then, by $\mathtt{K_whtoFS}$, $\vdash \B (x \approx y \vee \phi) \to (x \not \approx y \to \phi)$.
Hence, clearly $\vdash \B \phi \to (x \not \approx y \to \phi)$, i.e. $\vdash x \not \approx y \to (\B \phi \to \phi)$.

Then, we show that $\mathtt{MONO}$ is admissible in $\mathbf{S4.2}_\approx^{\KMF} \oplus x \not \approx y \to (\B \phi \to \phi)$.
Since we have $x \not \approx y \to (\B \phi \to \phi)$ for some fresh $y$,
by $\mathtt{K_{wh}toK}$, $\KMF^y (x \not \approx y) \to (\B \phi \to \phi)$.
Assume that $\vdash \phi \to \psi$.
We first prove that $\vdash \KMF^x \phi \to (\B \psi \to \psi)$.
On the one hand, $\vdash \neg \KMF^y (x \not \approx y) \to z \approx x$ (where $z$ is a fresh variable),
and thus $\vdash \K \phi[z/x] \wedge \neg \KMF^y (x \not \approx y) \to \phi$.
Then, since $\vdash \phi \to \psi$,
$\vdash \K \phi[z/x] \wedge \neg \KMF^y (x \not \approx y) \to \psi$,
and thus $\vdash \K \phi[z/x] \wedge \neg \KMF^y (x \not \approx y) \to (\B \psi \to \psi)$.
Then, by $\mathtt{K_{wh}toK}$ and $\mathtt{R}^{\KMF}$
$\vdash \KMF^x \phi \wedge \neg \KMF^y (x \not \approx y) \to (\B \psi \to \psi)$.
On the other hand, clearly $\vdash \KMF^x \phi \wedge \KMF^y (x \not \approx y) \to (\B \psi \to \psi)$.
Hence, $\vdash \KMF^x \phi \to (\B \psi \to \psi)$.
And since $\vdash \phi \to \psi$,
we also have $\vdash \K \phi \to \K \psi$.
Hence, by $\mathtt{FS\&KtoK_{wh}}$,
$\vdash \KMF^x \phi \wedge \K \phi \to \KMF^x \psi$,
and thus $\vdash \KMF^x \phi \to \KMF^x \psi$ by $\mathtt{K_{wh}toK}$.
\end{proof}

\end{document}